\documentclass{article}
\usepackage[utf8]{inputenc}
\usepackage[letterpaper, margin=1.0in]{geometry}
\usepackage{hyperref}
\usepackage{amsmath}
\usepackage{amsthm}
\usepackage{amsfonts}
\usepackage{amssymb}
\usepackage{braket}
\usepackage{graphicx}
\usepackage{physics}
\usepackage{authblk}

\newtheorem{theorem}{Theorem}[section]

\newtheorem{definition}{Definition}[section]

\newcommand{\eps}{\varepsilon}

\begin{document}

\title{Pretty Good Bounds on the worst-case Pretty Good Measurement}

\author{
  Sergio Escobar\thanks{UC Berkeley, Email: sesco99@berkeley.edu}
 \qquad
    Austin Pechan\thanks{UC Berkeley, Email: agp@berkeley.edu}
}
\date{}
\maketitle


\begin{abstract}
We derive a new lower bound on the success probability of the Pretty Good Measurement (PGM) for worst-case quantum state discrimination among $m$ pure states. Our bound is strictly tighter than the previously known Gram-matrix-based bound for $m\geq 4$. The proof adapts techniques from Barnum and Knill's analysis of the average-case PGM, applied here to the worst-case scenario. By comparing the PGM to the sequential measurement algorithm, we obtain a guarantee showing that, in the low-fidelity regime, the PGM's success probability decreases quadratically with respect to the maximum pairwise overlap, rather than linearly as in earlier bounds.
\end{abstract}

\section{Introduction}

The quantum state discrimination problem is a cornerstone of quantum information science, with profound implications for quantum communication, computing, and cryptography. Specifically in this paper, we consider the following problem: given a single copy of an unknown pure state $\ket{\psi}$ known to be one of $m$ possible pure states $\ket{v_1}, \ldots, \ket{v_m}$, the goal is to determine the correct index $i$ such that $\ket{\psi} = \ket{v_i}$. Unlike classical states, which can always be distinguished with certainty, quantum states may overlap non-trivially, making perfect discrimination impossible. This overlap reflects the probabilistic nature of quantum mechanics and highlights a fundamental departure from classical information theory.

The importance of the quantum state discrimination problem extends beyond its theoretical elegance; it underpins the reliability of quantum technologies. In quantum communication, accurate state identification ensures secure and efficient information transfer. In quantum cryptography, the impossibility of perfectly distinguishing non-orthogonal states is a cornerstone of protocol security. In quantum computing and information processing, state discrimination underlies algorithm output measurement, error syndrome extraction, and verification tasks. A detailed treatment of these connections can be found in Watrous's \textit{The Theory of Quantum Information}~\cite{watrous2018theory}.

Recent research has concentrated on establishing theoretical limits for quantum state discrimination, particularly in the average case (also known as the Bayesian case), where the state $\ket{\psi} = \ket{v_k}$ is prepared with some prior probability $p_k$. However, understanding the worst-case scenario, where the goal is to minimize the probability of failure across all possible states, is also valuable, especially for applications in adversarial environments or error-prone quantum devices. In this paper, we contribute to this ongoing exploration by refining lower bounds for the Pretty Good Measurement (PGM) in the worst-case quantum state discrimination problem. By building on established proof techniques and by leveraging known lower bounds for the sequential measurement algorithm (SMA), we derive a tighter lower bound on the PGM success probability.

The lower bound success probability we achieve for the PGM is worse than the sequential measurement algorithm, but in practice, the PGM is still very favorable. The SMA requires you to probe the same quantum state in up to $m$ distinct measurement rounds, each round demanding the state remain coherent while the measurement device is recalibrated. In contrast, for the PGM you implement a single, fixed POVM and then read out once. That single interaction slashes coherence time demands and eliminates run-to-run reconfiguration of the measurement apparatus, making the PGM better for noisy quantum devices. 
\section{Background}
\subsection{Pretty Good Measurement}
The worst-case, pure state, PGM introduced by Montanaro~\cite{montanaro2019prettysimpleboundsquantum} is defined as follows:
\begin{definition}[Worst-case PGM]
    Suppose we are discriminating between pure states $\{\ket{v_1}, \dots, \ket{v_m}\}$, then the worst-case PGM is defined as the POVM $\{E_i\}_i$ with measurement operators
\[E_i = S^{-1/2}\ketbra{v_i}S^{-1/2}\]
where
\[S = \sum_{i=1}^m \ketbra{v_i}.\]
\end{definition}

With this, we denote the success probability of the worst-case PGM, given a single copy of our state $\rho = \ketbra{v_k}$, as $P_\mathrm{PGM}$. From analyzing the Gram matrix of the states, Montanaro~\cite{montanaro2019prettysimpleboundsquantum} shows the following lower bound on this success probability:
\begin{equation}
    P_\mathrm{PGM} \geq 1 - mF
\end{equation}
where $F = \max_{i \neq j} \vert \braket{v_i}{v_j} \vert$.

\subsection{Sequential Measurement Algorithm}\label{seqmeasalg}
As introduced by Wilde~\cite{Wilde2013Sequential}, the sequential measurement algorithm (SMA) for pure states is defined as follows:
\begin{definition}[SMA]
    Given $m$ pure states $\ket{v_1},\ldots, \ket{v_m}$, we first define $$\Pi_i = \ketbra{v_i} \quad \text{and } \quad \overline{\Pi}_i = \mathbb{I} - \Pi_i  \quad  \forall i \in \{1, \ldots, m \}.$$ Now, given the state we want to discriminate, $\rho =\ketbra{v_k}$, measure it with $\{ \Pi_1, \overline{\Pi}_1 \}, \ldots, \{ \Pi_m, \overline{\Pi}_m \}$ sequentially. At any point, if we measure and get a $\Pi_i$ instead of a $\overline{\Pi}_i$ we return $\ket{v_i}$ as our guess. If we never measure a $\Pi_i$ after performing all $m$ measurements, then we pick one of the $m$ states uniformly at random and return it.
\end{definition}
We will call the success probability of this algorithm $P_\mathrm{SM}$, and we will break it down into two possibilities. If we are given the initial state we want to discriminate $\rho = \ketbra{v_k}$, then these parts are:
\begin{enumerate}
    \item $\Pr[\mathrm{Good}_1] := \Pr[\text{we see the sequence  }\overline{\Pi}_1 \cdots \overline{\Pi}_{k-1} \Pi_k]$
    \item $\Pr[\mathrm{Good}_2] := \Pr[\text{we see the sequence  $\overline{\Pi}_1 \cdots \overline{\Pi}_m$ and then randomly guess $\ketbra{v_k}$}]$
\end{enumerate}
With these, the success probability can be written as
$$ \Pr[\text{success}] = \Pr[\mathrm{Good}_1] + \Pr[\mathrm{Good}_2] $$

First, observe that
\[\Pr[\mathrm{Good}_2] = \frac{1}{m} \Tr(\prod_{i=1}^m \overline{\Pi}_i \ketbra{v_k}).\] 
Then, to lower bound $\Pr[\mathrm{Good}_1]$ we first need to define
\begin{align*}
    \epsilon_i &= \Tr(\Pi_i \ketbra{v_k}) = |\braket{v_i}{v_k}|^2 \leq F, \text{ for all $1\leq i \leq k-1$,}\\
    \epsilon_k &= \Tr(\overline{\Pi}_k \ketbra{v_k}) = 0
\end{align*}
Now, the Quantum Union Bound tells us that the probability that $\Pr[\mathrm{Good}_1]$ occurs is:
$$ \Pr[\mathrm{Good}_1] = 1 - \Pr[\mathrm{Bad}_1] \leq 1 - 4\sum_{i=1}^{k-1} \epsilon_i \leq 1 - 4(m-1)F^2$$
Finally, putting everything together, the probability of success of the sequential measurement algorithm is lower bounded by: 
\begin{equation}\label{eq:SM_lower_bound}
    P_\mathrm{SM} \geq 1 - 4(m-1)F^2 + \frac{1}{m} \Tr(\prod_{i=1}^m \overline{\Pi}_i \ketbra{v_k}).
\end{equation}

\section{Bounds on Worst-Case PGM}

In this section, we prove Theorem~\ref{worstcasePGM} improves upon the known lower bound of the worst-case PGM success probability. The proof follows the structure of the proof that $P_\mathrm{PGM}^\mathrm{avg} \geq P_\mathrm{OPT}^2$ presented by Barnum and Knill~\cite{barnum2000reversingquantumdynamicsnearoptimal}. Our proof is an adaptation of the argument shown there, applied to the worst-case setting.

\begin{theorem}\label{worstcasePGM}
    The success probability of the PGM for the worst-case discrimination problem using one copy of $\rho = \ketbra{v_i}$ is at least $\frac{(1-4(m-1)F^2)^2}{1+mF^2} = 1 - O(F^2)$.
\end{theorem}

\begin{proof}
    The Sequential Measurement Algorithm is given by measurement operators $\{ M_1, \dots, M_{m+1} \}$, where 
    \begin{align*}
        M_1 &= \Pi_1 \\
        M_{j} &= \prod_{i=1}^{j-1} (\mathbb{I} - \Pi_i)\Pi_j \prod_{i=1}^{j-1} (\mathbb{I} - \Pi_{j-i}) \ \ \ \text{for $2\leq j \leq m$}\\
        M_{m+1} &= \prod_{i=1}^m (\mathbb{I} - \Pi_i)
    \end{align*}

    Here, the $M_{m+1}$ term refers to the case where all sequential measurements fail, so we guess randomly. To verify the correctness of these $M_i$'s, given a state we want to discriminate $\ket{\psi} = \ket{v_k}$, if we output the guess ``$\ket{v_t}$" from the Sequential Measurement Algorithm, that means that the current state is
    \[\Pi_t \cdot (\mathbb{I} - \Pi_{t-1}) \cdots (\mathbb{I} - \Pi_1)\ket{v_k}.\]
    This happens with probability $$||\Pi_t \cdot (\mathbb{I} - \Pi_{t-1}) \cdots (\mathbb{I} - \Pi_1)\ket{v_k}||^2 = \bra{v_k} (\mathbb{I} - \Pi_1) \cdots (\mathbb{I} - \Pi_{t-1}) \cdot \Pi_t \cdot (\mathbb{I} - \Pi_{t-1}) \cdots (\mathbb{I} - \Pi_1)\ket{v_k}.$$
    Equivalently, if we are given the initial state $\ket{\psi} =\ket{v_k}$ and measure it with the POVM as described above, then the probability of seeing output $M_t$ (in which case we return "$\ket{v_t}$" as our guess) is 
    \begin{align*}
        \Tr(M_t \ketbra{v_k}) &= \bra{v_k} M_t \ket{v_k} \\
        &= \bra{v_k} (\mathbb{I} - \Pi_1) \cdots (\mathbb{I} - \Pi_{t-1}) \cdot \Pi_t \cdot (\mathbb{I} - \Pi_{t-1}) \cdots (\mathbb{I} - \Pi_1)\ket{v_k}.
    \end{align*}
    Therefore, these $M_i$'s are indeed the correct operators for the Sequential Measurement Algorithm.

    Using this POVM notation, we can write the worst case success probability for the sequential measurement algorithm as follows
    \begin{align*}
        P_\mathrm{SM} &= \min_i \Pr[\text{guess $i$} \mid i] \\
        &=\min_i \Pr[M_i \mid \ketbra{v_i}] + \frac{1}{m} \Pr[M_{m+1} \mid \ketbra{v_i}]\\
        &= \min_i \Tr(M_i \ketbra{v_i}) + \frac{1}{m} \Tr(M_{m+1} \ketbra{v_i}) \\
        &= \min_i \Tr(S^{1/4} M_i S^{1/4} S^{-1/4}\ketbra{v_i} S^{-1/4}) + \frac{1}{m} \Tr(M_{m+1} \ketbra{v_i})
    \end{align*}
    Combining this with eq. \ref{eq:SM_lower_bound} gives us that
    \begin{equation} \label{eq:intermediate_step}
        1- 4(m-1)F^2 \leq \min_i \Tr(S^{1/4} M_i S^{1/4} S^{-1/4}\ketbra{v_i} S^{-1/4})
    \end{equation}
    For simplicity, we will define $P_\mathrm{SM}' := 1- 4(m-1)F^2$, $A_i := S^{1/4} M_i S^{1/4}$, and $B_i := S^{-1/4}\ketbra{v_i} S^{-1/4}$. Now, by applying Cauchy-Schwarz and these new definitions to eq. \ref{eq:intermediate_step} gives:
    \begin{equation} \label{eq:CS_applied_to_SM_bound}
        P_\mathrm{SM}' \leq \min_i \Tr(A_i \cdot B_i) \leq \min_i \sqrt{\Tr(A_i^2)} \sqrt{\Tr(B_i^2)}. 
    \end{equation}
    We can now analyze each term of the equation. \ref{eq:CS_applied_to_SM_bound} individually. Starting with $A_i$ term:
    \begin{align*}
        \Tr(A_i^2) &= \Tr(S^{1/4} M_i S^{1/4} \cdot S^{1/4} M_i S^{1/4}) \\
        &= \Tr(S^{1/2}M_i S^{1/2} M_i) 
    \end{align*}
    Now, each $M_i \leq \mathbb{I}$ by definition, so
    \begin{align*}
        \Tr(A_i^2) \leq \Tr(S^{1/2}M_i S^{1/2}) = \Tr(S \cdot M_i).
    \end{align*}
    We know that $M_i = \prod_{j=1}^{i-1} (\mathbb{I} - \ketbra{v_j}) \cdot \ketbra{v_i} \cdot \prod_{j=1}^{i-1} (\mathbb{I} - \ketbra{v_{i-j}}) $, so by plugging this back in we get
    \begin{align*}
        \Tr(A_i^2) &\leq \Tr(S \cdot  \prod_{j=1}^{i-1} (\mathbb{I} - \ketbra{v_j}) \cdot \ketbra{v_i} \cdot \prod_{j=1}^{i-1} (\mathbb{I} - \ketbra{v_{i-j}}))\\
        &\leq \Tr(S \cdot  \prod_{j=1}^{i-1} (\mathbb{I} - \ketbra{v_j}) \cdot \ketbra{v_i}) \\
        &= \Tr(\ketbra{v_i} S \cdot  \prod_{j=1}^{i-1} (\mathbb{I} - \ketbra{v_j}))\\
        &\leq \Tr(\ketbra{v_i} S)
    \end{align*}
    Here we used the cyclic property of the trace and the fact that $\mathbb{I} - \ketbra{v_j} \leq \mathbb{I}$. So, the trace becomes
    \begin{align*}
        \Tr(A_i^2) &\leq \Tr(\ketbra{v_i} S) \\
        &= \bra{v_i} S\ket{v_i}\\
        &= \bra{v_i} \sum_{j=1}^m \ketbra{v_j} \ket{v_i} \\
        &= 1 + \sum_{j \neq i} |\braket{v_i}{v_j}|^2 \\
        &\leq 1 + mF^2
    \end{align*}
    Now, the bound on $P_\mathrm{SM}'$ becomes
    $$ P_\mathrm{SM}' \leq \min_i \sqrt{1+mF^2} \sqrt{\Tr(B_i^2)} = \sqrt{1+mF^2} \cdot \sqrt{\min_i\Tr(B_i^2)} $$
    Analyzing $\min_i \Tr(B_i^2)$ we find 
    \begin{align*}
        \min_i \Tr(B_i^2) &= \min_i \Tr(S^{-1/4}\ketbra{v_i} S^{-1/4} \cdot S^{-1/4}\ketbra{v_i} S^{-1/4}) \\
        &= \min_i \Tr(S^{-1/2} \ketbra{v_i} S^{-1/2} \ketbra{v_i}) \\
        &= \min \Tr(M_i^{PGM} \ketbra{v_i}) \\
        &= P_\mathrm{PGM}
    \end{align*}
    Putting everything together 
    \begin{align*}
        &P_\mathrm{SM}' \leq \sqrt{1+mF^2} \sqrt{P_\mathrm{PGM}}\\
        \implies & P_\mathrm{PGM} \geq \frac{(P_\mathrm{SM}')^2}{1+mF^2} = \frac{(1-4(m-1)F^2)^2}{1+mF^2}\qedhere
    \end{align*}
\end{proof}
This updated worst-case PGM lower bound outperforms the lower bound by Montanaro~\cite{montanaro2019prettysimpleboundsquantum} of $1-mF$, for the setting where $m\geq 4$. The full proof for this can be seen in App. \ref{sec:proof_of_improved_bound} or empirical results are shown in Fig. \ref{fig:improved_bound}.

\begin{figure}[h]
    \centering
    \includegraphics[width=0.8\linewidth]{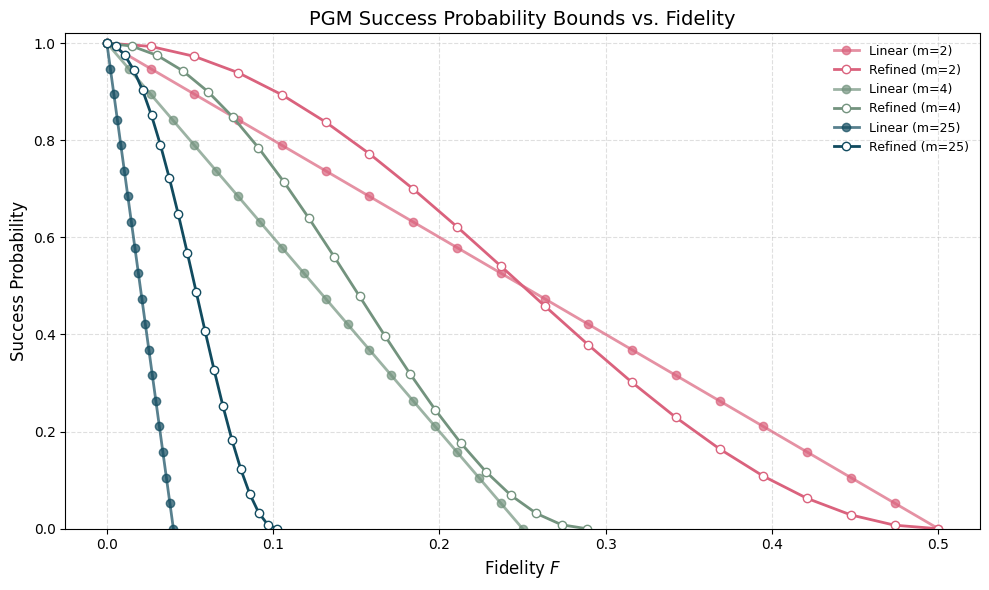}
    \caption{Plot of PGM success probabilities as a function of $F$ for multiple $m$ values. The linear terms are given by $1 - mF$ and the refined terms are given by $\frac{\left(1-4\left(m-1\right)F^2\right)^2}{1+mF^2}$. For $m\geq 4$, we see that our refined bound always outperforms the linear bound (as proved explicitly in App. \ref{sec:proof_of_improved_bound}).}
    \label{fig:improved_bound}
\end{figure}

\section{Conclusion}
In this paper, we prove an improved lower bound on the success probability of the Pretty Good Measurement (PGM) algorithm in the worst-case setting, when the number of states being discriminated is $\geq 4$. We did this by adapting the Barnum and Knill~\cite{barnum2000reversingquantumdynamicsnearoptimal} proof technique and comparing the PGM directly to the sequential measurement algorithm. We improve on previous results by showing that the success probability falls off quadratically rather than linearly in the low-fidelity regime.

Going forward, repeating this Barnum-Knill-esque argument of Theorem~\ref{worstcasePGM} with other measurement algorithms that do better than the sequential measurement algorithm (or even possibly getting a lower bound in terms of the optimal success probability) is of much interest. This could further improve the lower bound on the worst-case PGM. Additionally, defining and proving success bounds on a PGM for the unambiguous state discrimination problem is another avenue worth exploring. 

\newpage

%
%

\bibliographystyle{plain} 
\bibliography{biblio} 
\section*{Acknowledgements}
We thank John Wright for his helpful suggestions on this work.
\appendix
\section{Proof of  improved bound} \label{sec:proof_of_improved_bound}
In this section, we will prove that for $F>0$:
\begin{equation} \label{eq:inequality_for_improved_bound}
    \frac{\left(1-4\left(m-1\right)F^2\right)^2}{1+mF^2} > 1-mF, \quad \forall m\geq 4.
\end{equation}
We make the assumption that $F>0$, since if $F=0$, $P_{PGM} =1$ trivially.

To prove eq. \ref{eq:inequality_for_improved_bound}, we will equivalently show that the following function is greater than zero:
\begin{equation}
    g(F,m) = (1-4\left(m-1\right)F^2)^2 - (1+mF^2) \cdot (1-mF)
\end{equation}
Expanding $g(F,m)$, we get that
\begin{align*}
    g(F,m) &= \left( 1 - 8(m-1)F^2 + 16(m-1)^2F^4 \right) - \left( 1-mF + mF^2 - m^2F^3 \right) \\
    &= F\left( 16(m-1)^2F^3 +m^2F^2 -8(m-1)F - mF +m \right)
\end{align*}
Since we only want to show that $g(F,m)>0$, and we assume that $F>0$, we can equivalently show that $16(m-1)^2F^3 +m^2F^2 -8(m-1)F - mF +m > 0$. To do this, we define another function $h(F,m)$:
\begin{equation}
    h(F,m) = 16(m-1)^2F^3 +m^2F^2 -8(m-1)F - mF +m
\end{equation}
We will prove in two steps that $\forall F\in (0,1]$ and $m\geq 4$, that $h(F,m) >0$. First we will show that that $h(F,4) >0$, then we will show that $h(F,m> 4) \geq h(F,4)$. Combining these completes the proof by giving us
$$ h(F,m> 4) \geq h(F,4) > 0 \quad \forall F\in (0,1]. $$
\subsection{Proof $h(F,4) >0$}
Plugging in $m=4$ to $h()$ gives:
\begin{equation}
    h(F,4) =  144F^3 + 16F^2-28F+4.
\end{equation}
$h(F,4)$ only has one local extrema in $F\in (0,1]$. To find this, we set the derivative of $h$ with respect to $F$ equal to $0$.
\begin{align*}
    &\frac{d}{dF}h(F,4) = 432F^2+32F-28 = 0\\
    \implies& F = \frac{-32\pm \sqrt{32^2 + 4\cdot 144\cdot 28}}{2\cdot 432} \\
    \implies& F_{+} \approx 0.22 \in (0,1]
\end{align*}
The other solution is clearly outside of the domain as it is negative. By taking the second derivative, we get
$$ h''(F,4) = 864F+32 > 0 \quad \forall F\in (0,1] $$
So since $h''(F,4) >0$, the point $F_+$ is a local minimum, and since $h$ is continuous, $F_+$ is a global minimum in the domain $F\in (0,1]$. Plugging this point into $h$ gives us that
$$ h(F_+, 4) \approx 0.148 > 0. $$
Since $h(F_+, 4)$ is a global minimum in the range $F\in (0,1]$, for every $F\in (0,1]$, $h(F,4) \geq h(F_+, 4)>0$.
\subsection{Proof $h(F, m>4) \geq h(F,4)$}
To prove this, we will use the fact that if a function is continuous on $[a,b]$ and differentiable on $(a,b)$, and its derivative is greater than zero on $(a,b)$, then the function is strictly increasing on $[a,b]$ (by the Mean Value Theorem). Our function $h(\cdot)$ is continuous and differentiable everywhere, and since we are concerned with $m \in [4,\infty)$, it suffices to take any $m>4$ and apply this fact to the finite interval $[4,m]$. Thus, it is enough to show that
$$
\frac{\partial}{\partial m} h(F,m) > 0 
\quad\text{for all } m>4 \text{ and } F \in (0,1],
$$
which will imply
$$
h(F,m) \ge h(F,4) \quad \text{for all } m \ge 4.
$$

To start, 
\begin{align*}
    \frac{d}{dm}h(F, m) &= 32F^3(m-1) + 2mF^2-9F+1 \\
    &\geq 32F^3(4-1) + 2\cdot 4\cdot F^2-9F+1 &\text{(since $m\geq 4$)}\\
    &= 96F^3+8F^2-9F+1
\end{align*}
To show that this is greater than $0$, we first let $p(F) = 96F^3+8F^2-9F+1$. Taking the second derivative, it is clear to see that $p''(F)>0$ $\forall F\in (0,1]$. With this, we know that the local minimum in the range $F\in (0,1]$ must also be a global minimum in the same range.
\begin{align*}
    p'(F) = 288F^2 + 16F-9
\end{align*}
Setting $p'(F) = 0$, gives us that the only local extrema for $F\in (0,1]$ is
$$ F_+ = \frac{-16 + \sqrt{16^2+4\cdot 288 \cdot 9}}{2\cdot 288} \approx 0.151 $$
With this, $p(F_+) \approx 0.154 > 0$. So, for any choice of $F$ we have that $\frac{d}{dm} h(F,m) >0$. Therefore, $h(F,m)$ is increasing on the interval $[4, \infty)$ making $h(F, m>4) \geq h(F,4)$. We actually show a stronger argument which is for all $m\in [4,\infty)$ and any $\eps > 0$ that $h(F,m+\eps)>h(F,m)$.

\end{document}